\documentclass[10pt,conference,letter]{IEEEtran}

\usepackage{times}

\usepackage{amssymb}
\usepackage{amsthm}
\usepackage{amsxtra}
\usepackage{amsmath}
\usepackage{array}
\usepackage{algpseudocode}
\usepackage{algorithm}

\usepackage{verbatim}
\usepackage{epsfig}
\usepackage{setspace}
\usepackage{caption}
\usepackage{subcaption}
\usepackage{breqn}

\allowdisplaybreaks

\newtheorem{theorem}{Theorem}
\newtheorem{definition}{Definition}
\newtheorem{lemma}[theorem]{Lemma}
\newtheorem{proposition}[theorem]{Proposition}

\newtheorem{example}{Example}
\newtheorem{claim}{Claim}
\newtheorem*{remark*}{Remarks}

\def\CI{\text{CI}}

\def\CO{\text{CO}}
\def\SK{\text{SK}}

\def\cB{\mathcal{B}}
\def\cE{\mathcal{E}}
\def\cF{\mathcal{F}}

\def\cM{\mathcal{M}}
\def\cP{\mathcal{P}}
\def\cS{\mathcal{S}}
\def\cV{\mathcal{V}}
\def\cX{\mathcal{X}}

\def\BL{\textbf{L}}

\def\cH{\mathcal{H}}

\def\N{\mathbb{N}}
\usepackage[
  bookmarks=false,
  pdfpagelabels=false,
  hyperfootnotes=false,
  hyperindex=false,
  pageanchor=false,
  colorlinks,
  linkcolor=black,
  citecolor=black,
  filecolor=black,
  urlcolor=blue
]{hyperref}

\title{The Communication Complexity of Achieving \\ SK Capacity in a Class of PIN Models}
\author{
\IEEEauthorblockN{Manuj Mukherjee$^\dag$} \and \IEEEauthorblockN{Navin Kashyap$^\dag$}
}

\begin{document}

\maketitle

\renewcommand{\thefootnote}{}
\footnotetext{
\noindent $^\dag$M.\ Mukherjee and N.\ Kashyap are with the 
Department of Electrical Communication Engineering, 
Indian Institute of Science, Bangalore. Email: \{manuj,nkashyap\}@ece.iisc.ernet.in.

\smallskip
}

\begin{abstract}
The communication complexity of achieving secret key (SK) capacity in the multiterminal source model of Csisz{\'a}r and Narayan is the minimum rate of public communication required to generate a maximal-rate SK. It is well known that the minimum rate of communication for omniscience, denoted by $R_{\CO}$, is an upper bound on the communication complexity, denoted by $R_{\SK}$. A source model for which this upper bound is tight is called $R_{\SK}$-maximal. In this paper, we establish a sufficient condition for $R_{\SK}$-maximality within the class of pairwise independent network (PIN) models defined on hypergraphs. This allows us to compute $R_{\SK}$ exactly within the class of PIN models satisfying this condition. On the other hand, we also provide a counterexample that shows that our condition does not in general guarantee $R_{\SK}$-maximality for sources beyond PIN models.
\end{abstract}

\renewcommand{\thefootnote}{\arabic{footnote}}

\section{Introduction} \label{sec:intro}
Csisz{\'a}r and Narayan \cite{CN04} introduced the problem of secret key (SK) generation within the multiterminal i.i.d.\ source model. In this model, there are multiple terminals, each of which observes a distinct component of a source of correlated randomness. The goal is for the terminals to agree on a shared SK via communication over an insecure noiseless public channel. The SK is to be secured from passive eavesdroppers with access to the public channel. The maximum rate of such an SK, i.e. the \emph{SK capacity}, was characterized in \cite{CN04}, and a protocol for attaining SK capacity was given, which involved communication for \emph{omniscience}, i.e., all terminals recovering the entire information of all the other terminals. However, it was pointed out (see remark following Theorem 1 in \cite{CN04}) that omniscience is not always necessary for achieving SK capacity. A question that naturally arises is the following (see \cite[Section VI]{CN04} and \cite[Section V]{NN10}): what is the minimum rate of public communication required to achieve SK capacity? We call this minimum rate of public communication the \emph{communication complexity}\footnote{Our use of ``communication complexity'' differs from the use prevalent in the theoretical computer science literature where, following \cite{Yao79}, it refers to the total amount of communication, in bits, required to perform some distributed computation.} of achieving SK capacity, and denote it by $R_{\SK}$. The protocol from \cite{CN04} shows that $R_{\SK}$ is upper bounded by the minimum rate of public communication required for omniscience, denoted by $R_{\CO}$. We refer to sources for which this upper bound is tight as \emph{$R_{\SK}$-maximal}.

There have been a few attempts at characterizing $R_{\SK}$. In \cite[Theorem 3]{Tyagi13} Tyagi has completely characterized the communication complexity for two terminals in terms of an \emph{interactive common information}, a type of Wyner common information \cite{Wyner75}. Our previous work \cite{MK14} involved extension of Tyagi's results to the case of $m>2$ terminals. Specifically, we gave a lower bound \cite[Theorem~2]{MK14} on the communication complexity using a multiterminal variant of Tyagi's interactive common information. We were able to evaluate this lower bound only in the very special case of a complete graph pairwise independent network (PIN) model in which we additionally imposed linearity restrictions on the public communication allowed \cite[Theorem~6]{MK14}.

A different approach to analyzing $R_{\SK}$ can be found in \cite{CHISIT14},\cite{CH14}. These follow up on the work in \cite{CW14}, which studied \emph{one-shot} SK generation (i.e., each component of the source just gives out one symbol instead of a sequence of i.i.d.\ symbols) in a hypergraph PIN model, and evaluated the corresponding one-shot SK capacity \cite[Theorem 6]{CW14}. This result also used communication for omniscience for attaining the one-shot SK capacity, but did not address the issue of communication complexity. This isssue was addressed in the subsequent work \cite{CH14}, which characterized the communication complexity of achieving one-shot SK capacity under linearity restrictions on the communication. The characterization was in terms of ``minimum connected dominating edge sets'' of hypergraphs \cite[Theorem~11]{CH14}. While the general problem of determining the unrestricted communication complexity was left open, it was shown that removing the linearity restriction can strictly reduce the communication complexity in some cases \cite[Theorem 4]{CH14}.

The main contribution of this work is the identification of a sufficient condition under which a certain class of hypergraph PIN models (of which the simple graph PIN models of \cite{NN10} form a subclass) can be shown to be $R_{\SK}$-maximal. Thus, for this class, we have $R_{\SK} = R_{\CO}$, and the latter can be explicitly computed in terms of the parameters of the underlying hypergraph. This yields the first explicit computation of the (unrestricted) communication complexity $R_{\SK}$ for a multiterminal source model with more than two terminals.  This greatly extends our earlier results from \cite{MK14}, and also, in a sense, partially extends the one-shot results of \cite{CH14} to the i.i.d.\ source sequence model. However, it is also shown via a counterexample that our condition does not guarantee $R_{\SK}$-maximality for sources beyond the PIN model.

The rest of the paper is structured as follows. Section~\ref{sec:prelim} presents the required definitions and notation. Section~\ref{sec:uh} identifies a class of hypergraph PIN models which are $R_{\SK}$-maximal. Section~\ref{sec:counterex} shows using a counterexample that the results of Section~\ref{sec:uh} do not extend to a general multiterminal setting. The paper concludes with some remarks in Section~\ref{sec:conc}.

\section{Preliminaries} \label{sec:prelim}
We will follow the notation and description of \cite{MK14}. Throughout, we use $\N$ to denote the set of positive integers. Consider a set of $m\geq2$ terminals denoted by $\mathcal{M}=\{1,2, \ldots, m\}$. Each terminal $i \in \mathcal{M}$ observes $n$ i.i.d.\ repetitions of a random variable $X_i$ taking values in a finite set $\mathcal{X}_i$. The $n$ i.i.d.\ copies of the random variable are denoted by $X_i^n$. The random variables $X_1,X_2,\ldots,X_m$ need not be independent.
For any subset $A\subseteq \mathcal{M}$, $X_A$ and $X_A^n$ denote the collections of random variables $(X_i:i \in A)$ and $(X_i^n: i \in A)$, respectively. The terminals communicate through a noiseless public channel, 
any communication sent through which is accessible to all terminals and to potential eavesdroppers as well.
An \emph{interactive communication} is a communication $\textbf{f}=(f_1,f_2,\cdots,f_r)$ with finitely many transmissions $f_j$, in which any transmission sent by the $i$th terminal is a deterministic function of $X_i^n$ and all the previous communication, i.e., if terminal $i$ transmits $f_j$, then $f_j$ is a function only of $X_i^n$ and $f_1,\ldots,f_{j-1}$.   We denote the random variable associated with $\textbf{f}$ by $\textbf{F}$; the support of $\textbf{F}$ is a finite set $\cF$. The rate of the communication $\textbf{F}$ is defined as $\frac{1}{n}\log|\cF|$. Note that $\textbf{f}$, $\textbf{F}$ and $\cF$ implicitly depend on $n$.

\begin{definition}
\label{def:CR}
A \emph{common randomness (CR)} obtained from an interactive communication $\textbf{F}$ is a sequence of random variables $\textbf{J}^{(n)}$, $n\in\N$, which are functions of $X_{\mathcal{M}}^n$, such that for any $0<\epsilon<1$ and for all sufficiently large $n$, there exist $J_i=J_i(X_i^n,\textbf{F})$, $i = 1,2,\ldots,m$, satisfying $\text{Pr}[J_1=J_2=\cdots=J_m=\textbf{J}^{(n)}] \geq 1-\epsilon$.
\end{definition}

When $\textbf{J}^{(n)}=X_{\cM}^n$ we say that the terminals in $\cM$ have attained \emph{omniscience}. The communication $\textbf{F}$ which achieves this is called a \emph{communication for omniscience}. We denote the minimum rate of communication for omniscience by $R_{\CO}$.

\begin{definition}
\label{def:SK}
A real number $R\geq 0$ is an \emph{achievable SK rate} if there exists a CR $\textbf{K}^{(n)}$, $n \in \N$, obtained from an interactive communication $\textbf{F}$ satisfying, for any $\epsilon > 0$ and for all sufficiently large $n$, $I(\textbf{K}^{(n)};\textbf{F})\leq \epsilon$ and $\frac{1}{n}H(\textbf{K}^{(n)}) \geq R-\epsilon$. The \emph{SK capacity} is defined to be the supremum among all achievable rates.  The CR $\textbf{K}^{(n)}$ is called a \emph{secret key (SK)}. 
\end{definition}

From now on, we will drop the superscript $(n)$ from both $\textbf{J}^{(n)}$ and $\textbf{K}^{(n)}$ to keep the notation simple. 

The SK capacity can be expressed as \cite[Section~V]{CN04}, \cite{CZ10}
\begin{equation}
\textbf{I}(X_{\mathcal{M}}) \triangleq H(X_{\mathcal{M}})-\max_{\lambda \in \Lambda} \sum_{B \in \mathcal{B}} \lambda_B H(X_B| X_{B^c}) \label{skcapacity}
\end{equation}
where $\mathcal{B}$ is the set of non-empty, proper subsets of $\mathcal{M}$, and $\lambda=(\lambda_B: B\in \mathcal{B})\in \Lambda$ iff $\lambda_B\geq 0$ for all $B\in \mathcal{B}$ and for all $i\in \mathcal{M}$, $\sum_{B:i\in B}\lambda_B=1$. It is a fact that $\textbf{I}(X_{\mathcal{M}}) \ge 0$ \cite[Proposition~II]{MT10}. Other equivalent characterizations of $\textbf{I}(X_{\cM})$ exist in literature. Theorem 1 of \cite{CN04} shows that
\begin{equation}
\textbf{I}(X_{\cM})=H(X_{\cM})-R_{\CO}.\label{omni}
\end{equation}
Theorem 1.1 of \cite{CZ10} and Theorem 2.1 of \cite{Chan14} provides yet another characterization of $\textbf{I}(X_{\cM})$. Define $\Delta(\cP)\triangleq \frac{1}{|\cP|-1} \left[\sum_{A \in \cP} H(X_A) - H(X_{\cM}) \right]$. Then,
\begin{equation}
\textbf{I}(X_{\cM})=\min_{\cP}\Delta(\cP) 
\label{eq:I}
\end{equation}
the minimum being taken over all partitions $\cP=\{A_1,A_2,\cdots,A_{\ell}\}$ of $\cM$, of size $\ell \ge 2$. The partition $\bigl\{\{1\},\{2\},\ldots,\{m\}\bigr\}$ consisting of $m$ singleton cells will play a special role in the later sections of this paper; we call this the \emph{singleton partition} and denote it by $\cS$. The sources where \emph{$\cS$ is a minimizer for \eqref{eq:I}} will henceforth be refered to as \emph{Type $\cS$ sources}. The following proposition from \cite{MKS14} gives us an algorithm to verify whether a source is Type $\cS$. For any $B\subsetneq\cM$ with $B=\{b_1,b_2,\cdots,b_{|B|}\}$ denote by $\cP_B$ the partition $\cP_B=\{\{b_1\},\{b_2\},\cdots,\{b_{|B|}\},B^c\}$. Then we have
\begin{proposition}
\cite[Proposition 7]{MKS14}
For $m \ge 3$, let $\Omega = \{B \subset [m]: 1 \le |B| \, \le m-2\}$. The singleton partition $\cS$ is \\
\emph{(a)} a minimizer for $\mathbf{I}(X_{[m]})$ iff $\Delta(\cS) \le \Delta(\cP_B)$ $\forall\,B \in \Omega$; \\
\emph{(b)} the unique minimizer for $\mathbf{I}(X_{[m]})$ iff $\Delta(\cS) < \Delta(\cP_B)$ $\forall\,B \in \Omega$.
\label{prop:min}
\end{proposition}
A better (strongly polynomial-time) algorithm to calculate the minimizing partition of \eqref{eq:I} has been described in \cite{Chan14}. However, Proposition~\ref{prop:min} above is more suited for the purposes of this paper.

We are now in a position to make the notion of communication complexity rigorous.

\begin{definition}
\label{def:RSKr}
A real number $R\geq 0$ is said to be an \emph{achievable rate of interactive communication for maximal-rate SK} if for all $\epsilon > 0$ and for all sufficiently large $n$, there exist \emph{(i)}~an interactive communication $\textbf{F}$ satisfying $\frac{1}{n}\log|\cF| \; \leq R+\epsilon$, and \emph{(ii)}~an SK $\textbf{K}$ obtained from $\textbf{F}$ such that $\frac{1}{n}H(\textbf{K})\geq \textbf{I}(X_{\mathcal{M}})-\epsilon$.

We denote the infimum among all such achievable rates by $R_{\SK}$.
\end{definition}

The proof of Theorem~1 in \cite{CN04} shows that there exists an interactive communication $\textbf{F}$ that enables omniscience at all terminals and from which a maximal-rate SK can be obtained. Therefore, we have $R_{\SK}\leq R_{\CO}< \infty$.

In \cite{MK14} the communication complexity was lower bounded using extensions of proof techniques developed in \cite{Tyagi13}. The lower bound involves a quantity called the interactive common information rate, a special case of the Wyner common information rate \cite{Wyner75} extended to a multiterminal setting. We will now define formally what these quantities are. In order to do so we need the following extension of the definition of $\textbf{I}(X_{\mathcal{M}})$ given in \eqref{skcapacity}: for any random variable $\textbf{L}$, and any $n \in \N$, we define
\begin{equation}
\textbf{I}(X_{\mathcal{M}}^n | \textbf{L}) \triangleq \max_{\lambda\in\Lambda^*} \left[H(X_{\mathcal{M}}^n | \textbf{L})-\sum_{B \in \mathcal{B}} \lambda_B H(X_B^n|X_{B^c}^n,\textbf{L})\right], \label{cmi}
\end{equation}
where $\Lambda^* \subset \Lambda$ is the set constituting of optimal $\lambda \in \Lambda$ for the linear program in the definition of $\textbf{I}(X_{\mathcal{M}})$ in \eqref{skcapacity}.\footnote{The maximization carried out in \eqref{cmi} was not originally present in \cite{MK14}. The maximization has been brought in here to make the quantity $\textbf{I}(X_{\mathcal{M}}|\BL)$ well defined. It can be easily seen that under this modified definition the results of \cite{MK14} are still valid.} It follows from Proposition~II in \cite{MT10} that $\textbf{I}(X_{\mathcal{M}}^n | \textbf{L}) \ge 0$. Also, note that $\textbf{I}(X_{\cM}^n) = n\textbf{I}(X_{\cM})$. 

\begin{definition}
\label{def:CI}
A \emph{(multiterminal) Wyner common information ($\CI_W$)} for $X_{\mathcal{M}}$ is a sequence of finite-valued functions $\textbf{L}^{(n)}=\textbf{L}^{(n)}(X_{\mathcal{M}}^n)$ such that $\frac{1}{n}\textbf{I}(X_{\mathcal{M}}^n|\textbf{L}^{(n)}) \to 0$ as $n \to \infty$. An \emph{interactive common information ($\CI$)} for $X_{\mathcal{M}}$ is a Wyner common information of the form $\textbf{L}^{(n)} = (\textbf{J},\textbf{F})$, where $\textbf{F}$ is an interactive communication and $\textbf{J}$ is a CR obtained from $\textbf{F}$. 
\end{definition}

Again, we shall drop the superscript $(n)$ from $\textbf{L}^{(n)}$ for notational simplicity. Wyner common informations $\textbf{L}$ do exist: for example, the identity map $\textbf{L}=X_{\mathcal{M}}^n$ is a $\CI_W$. 
To see that $\CI$s $(\textbf{J},\textbf{F})$ also exist, observe that $\textbf{J}=X_{\mathcal{M}}^n$ and a communication $\textbf{F}$ enabling omniscience constitute a $\CI_W$, and hence, a $\CI$.

\begin{definition}
\label{def:CIrate}
A real number $R\geq 0$ is an \emph{achievable $\CI_W$ (resp.\ $\CI$) rate} if there exists a $\CI_W$ $\textbf{L}$ (resp.\ a $\CI$ $\textbf{L} = (\textbf{J},\textbf{F})$) such that for all $\epsilon > 0$, we have
$\frac{1}{n}H(\textbf{L})\leq R+\epsilon$ for all sufficiently large $n$. 

We denote the infimum among all achievable $\CI_W$ (resp.\ $\CI$) rates by $\CI_W(X_{\mathcal{M}})$ (resp. \ $\CI(X_{\mathcal{M}})$). 
\end{definition}

To ensure that $\CI(X_{\cM})<\infty$, existence of a $(\textbf{J},\textbf{F})$ pair which is a $\CI_W$ is needed. Such a pair indeed exists, as the proof of \cite[Theorem 1]{CN04} shows that there exists an interactive communication $\textbf{F}$ from which a CR $\textbf{J}=X_{\cM}^n$ is obtained, with $\BL=(\textbf{J},\textbf{F})$ being a $\CI_W$, as discussed after Definition~\ref{def:CI}.

The proposition below records the relationships between some of the information-theoretic quantities defined so far. 

\begin{proposition} \cite[Proposition 1]{MK14} For any source $X_{\cM}^n$, we have 
$H(X_{\mathcal{M}}) \ge \CI(X_{\mathcal{M}})\geq \CI_W(X_{\mathcal{M}})\geq \textbf{I}(X_{\mathcal{M}})$.
\label{prop:ineqs}
\end{proposition}

We conclude this section by stating the lower bound on communication complexity as derived in \cite{MK14}:

\begin{theorem}\cite[Theorem 2]{MK14}
$$
R_{\SK}\geq \CI(X_{\mathcal{M}})-\textbf{I}(X_{\mathcal{M}}).
$$
\label{th:commcomp}
\end{theorem}
\vspace*{-10pt} 
By Proposition~\ref{prop:ineqs}, the lower bound above is non-negative.

\section{$R_{\SK}$-maximality in uniform hypergraph PIN models} \label{sec:uh}

This section contains the main result of this work. First we will quickly introduce the hypergraph PIN model. The model is defined on an underlying hypergraph $\cH = (\cV,\cE)$ with $\cV = \cM$, the set of $m$ terminals of the model, and $\cE$ being a collection of hyperedges, i.e., subsets of $\cV$. For $n \in \N$, define $\cH^{(n)}$ to be the multi-hypergraph $(\cV,\cE^{(n)})$, where $\cE^{(n)}$ is the multiset of hyperedges formed by taking $n$ copies of each hyperedge of $\cH$. Associated with each hyperedge $e \in \cE^{(n)}$ is a Bernoulli$(1/2)$ random variable $\xi_e$; the $\xi_e$s associated with distinct hyperedges in $\cE^{(n)}$ are independent. With this, the random variables $X_i^n$, for $i\in\mathcal{M}$, are defined as $X_i^n=(\xi_e$ : $e\in\mathcal{E}^{(n)}$ and $i\in e$). When every $e\in\cE$ satisfies $|e|=t$, we call $\cH$ a \emph{$t$-uniform hypergraph}. We will show that any Type $\cS$ uniform hypergraph PIN model is $R_{\SK}$-maximal.

\begin{theorem}
For a Type $\cS$ PIN model defined on an underlying $t$-uniform hypergraph $\cH = (\cV,\cE)$, we have $\CI(X_{\cM})=\CI_W(X_{\cM})=H(X_{\cM})$, and hence, $R_{\SK}=R_{\CO} = \frac{m-t}{m-1} |\cE|$.
\label{th:uh}
\end{theorem}

The proof will require two technical lemmas which we state below. The first lemma identifies a $\lambda\in\Lambda^*$ when a source is Type $\cS$.

\begin{lemma}
Let the singleton partition $\cS$ be a minimizer for \eqref{eq:I}. Define $\tilde{\lambda} = (\tilde{\lambda}_B: B \in \cB)$ such that $\tilde{\lambda}_B = \frac{1}{m-1}$ whenever $|B| = m-1$, and $\tilde{\lambda}_B = 0$ otherwise. Then $\tilde{\lambda}\in\Lambda^*$.
\label{lem:singlam}
\end{lemma}

\begin{IEEEproof}
Observe that $\tilde{\lambda}\in\Lambda$. Putting $\lambda=\tilde{\lambda}$ in \eqref{skcapacity} we have $H(X_{\cM})-\sum_{B\in\cB}\tilde{\lambda}_BH(X_B| X_{B^c})=\Delta(\cS)=\textbf{I}(X_{\cM})$, as $\cS$ is a minimizer in \eqref{eq:I}. Thus $\tilde{\lambda}$ is optimal, i.e., $\tilde{\lambda}\in\Lambda^*$.
\end{IEEEproof}

\begin{lemma}
\label{lem:mi}
For any \emph{$t$-uniform hypergraph} PIN model and any function $\textbf{L}$ of $X_{\mathcal{M}}^n$ we have:
\begin{equation}
\sum_{i=1}^m I(X_i^n;\textbf{L}) \leq tH(\textbf{L}). \label{eq:mi}
\end{equation}
\end{lemma}

The lengthy proof of this lemma is deferred to the Appendix. We now proceed to prove Theorem \ref{th:uh}.

\begin{IEEEproof}[Proof of Theorem \ref{th:uh}]
For any Type $\cS$ source $\cX_{\cM}$, we have 
\begin{equation}
\textbf{I}(X_{\mathcal{M}}^n | \textbf{L}) \geq H(X_{\mathcal{M}}^n | \textbf{L})- \frac{1}{m-1} \sum_{i=1}^m H(X_{\cM \setminus \{i\}}^n|X_{i}^n,\textbf{L})  \label{cmi_uh}
\end{equation}
where \eqref{cmi_uh} follows from \eqref{cmi} and Lemma \ref{lem:singlam}. Now assume that $\cX_{\cM}$ arises from a PIN model defined on a $t$-uniform hypergraph $\cH = (\cV,\cE)$, and consider any function $\textbf{L}$ of $X_{\mathcal{M}}^n$. This allows us further simplification of \eqref{cmi_uh}:
\begin{align}
\textbf{I}(X_{\mathcal{M}}^n | \textbf{L}) 
   & \geq H(X_{\mathcal{M}}^n) - H(\textbf{L}) \notag \\
   & \hspace{1.2em} - \frac{1}{m-1} \sum_{i=1}^m \left[H(X_{\cM}^n) - H(X_i^n) - H(\textbf{L}|X_i^n)\right] \notag \\
   & = \frac{n(t-1)|\cE|}{m-1} - H(\textbf{L}) + \frac{1}{m-1} \sum_{i=1}^m H(\textbf{L}|X_i^n) \label{cmi_uh2}\\
   & = \frac{n(t-1)|\cE|}{m-1} - \frac{1}{m-1} \left[ \sum_{i=1}^m I(X_i^n;\BL)-H(\BL)\right] \notag \\
   & = \frac{n(t-1)}{m-1}\left(|\cE|-\frac{1}{n}H(\BL)\right) \notag \\
   & \hspace{1.2em}- \frac{1}{m-1} \left[ \sum_{i=1}^m I(X_i^n;\BL)-tH(\BL)\right] \notag \\
   & \geq \frac{n(t-1)}{m-1}\left(|\cE|-\frac{1}{n}H(\BL)\right), \label{cmi:uh3}
\end{align}
the equality \eqref{cmi_uh2} using the facts that $H(X_{\cM}^n) = n|\cE|$ and $\sum_{i=1}^m H(X_i^n) = nt|\cE|$, and \eqref{cmi:uh3} following from Lemma~\ref{lem:mi}.

We will now compute $\CI(X_{\cM})$ using Proposition \ref{prop:ineqs}. The upper bound gives us $\CI(X_{\cM}) \le |\cE|$, as $H(X_{\cM})=|\cE|$. For the lower bound, let $\textbf{L}$ be any $\CI_W$ so that for any $\epsilon > 0$, we have $\frac{1}{n} \textbf{I}(X_{\mathcal{M}}^n | \textbf{L}) < \frac{(t-1)\epsilon}{(m-1)}$ for all sufficiently large $n$. The bound in \eqref{cmi:uh3} thus yields $\frac{1}{n}H(\BL)> |\cE|-\epsilon$ for all sufficiently large $n$. Hence, it follows that $\CI_W(X_{\cM}) \ge |\cE|$. From the upper and lower bounds in Proposition \ref{prop:ineqs}, we now obtain $CI_W(X_{\cM})=CI(X_{\cM})=H(X_{\cM})$.

Now from Theorem \ref{th:commcomp} we have $R_{\SK}\geq CI(X_{\cM})-\textbf{I}(X_{\cM})$. Hence we have
\begin{gather}
R_{\SK} \ge |\cE|- \textbf{I}(X_{\cM})= H(X_{\cM})-\textbf{I}(X_{\cM})= R_{\CO}, \label{cmi:uh4}
\end{gather}
where the last equality is from \eqref{omni}. But we also have $R_{\SK}\leq R_{\CO}$, as pointed out in Section \ref{sec:prelim}, which proves that $R_{\SK} = R_{\CO}$. 

To obtain the exact expression for $R_{\CO}$, we note that by \eqref{omni} and \eqref{eq:I}, $R_{\CO} = H(X_{\cM}) - \Delta(\cS) = \frac{m}{m-1} H(X_{\cM}) - \frac{1}{m-1} \sum_{i=1}^m H(X_i)$. This simplifies to the expression stated in the theorem using the facts (already mentioned above) that $H(X_{\cM}) = |\cE|$ and $\sum_{i=1}^m H(X_i) = t |\cE|$.
\end{IEEEproof} 

We will now show that there indeed exist Type $\cS$ $t$-uniform hypergraph PIN models. Call $K_{m,t}=(\cV,\cE)$ a \emph{complete $t$-uniform hypergraph} on $m$ vertices when $e\subset\cV$ is contained in $\cE$ iff $|e|=t$. Using Proposition \ref{prop:min} we show that complete $t$-uniform hypergraph PIN models are Type $\cS$.

\begin{lemma}
\label{lem:comuh}
Complete $t$-uniform hypergraph PIN models are Type $\cS$.
\end{lemma}

\begin{proof}
Fix a set $B\subsetneq\cM$ with $|B|\leq m-2$. We calculate $\Delta(P_B)$, where $P_B$ is defined as in Proposition~\ref{prop:min}, and will show that $\Delta(P_B)>\Delta(\cS)$.  For $K_{m,t}$ we have, $H(X_i)=\binom{m-1}{t-1}$ and $H(X_{\cM})=\binom{m}{t}$ and therefore $\Delta(\cS)=\frac{t-1}{m-1}\binom{m}{t}$. To evaluate $\Delta(P_B)$, note that $H(X_{B^c})$ is the total number of hyperedges in $\cE$ which contain at least one terminal from $B^c$. Observe that if ${|B|}\geq t$ we have $H(X_{B^c})=\binom{m}{t}-\binom{|B|}{t}$. Otherwise, we have $H(X_{B^c})=\binom{m}{t}$.

So first consider ${|B|}\geq t$. Under this condition we see that 
\begin{align}
\Delta(P_B)& =\frac{1}{|B|}\left(\sum_{i\in B} H(X_i)+H(X_{B^c})-H(X_{\cM})\right)\notag\\
                  &=\binom{m-1}{t-1}-\frac{1}{|B|}\binom{|B|}{t}. \notag
\end{align}
Thus,
\begin{align}
\Delta(P_B)-\Delta(\cS) & = \binom{m-1}{t-1}-\frac{1}{|B|}\binom{|B|}{t}-\frac{t-1}{m-1}\binom{m}{t} \label{comuh:1}\\
                                      & = \frac{1}{t}\biggl[\frac{(m-1)!\ t}{(m-t)!\ (t-1)!}\nonumber \\
                                      & \hspace{1.2em}-\frac{m!}{(t-2)!\ (m-t)!\ (m-1)}\notag\\
                                      &\hspace{1.2em}-\binom{|B|-1}{t-1}\biggr] \nonumber\\
                                      & = \frac{1}{t}\biggl[\frac{(m-1)!}{(t-2)!\ (m-t)!}\left(\frac{t}{t-1}-\frac{m}{m-1}\right)\notag\\
                                      & \hspace{1.2em}-\binom{|B|-1}{t-1}\biggr] \nonumber\\
                                      & = \frac{1}{t}\left[\binom{m-2}{t-1}-\binom{|B|-1}{t-1}\right] \label{comuh:2}\\
                                      & \geq 0 \label{case1} 
\end{align}
where \eqref{case1} holds as ${|B|}\leq m-2$.

Next consider ${|B|}<t$. Under this condition we have
\begin{align}
\Delta(P_B)&=\frac{1}{|B|}\left(\sum_{i\in B} H(X_i)+H(X_{B^c})-H(X_{\cM})\right)\notag\\
                  &=\binom{m-1}{t-1}. \notag
\end{align}
Thus, using \eqref{comuh:1} and \eqref{comuh:2} we have
\begin{align}
\Delta(P_B)-\Delta(\cS)& =\binom{m-1}{t-1}-\frac{t-1}{m-1}\binom{m}{t}\notag\\
                                     & =\frac{1}{t}\binom{m-2}{t-1}\notag \\ 
                                     & \geq 0. \label{case2}
\end{align}

Using Proposition \ref{prop:min}, \eqref{case1} and \eqref{case2}, we have the result.
\end{proof}

\begin{remark*}
There is in fact a broad class of ordinary graph ($t=2$) PIN models which are Type $\cS$. Corollary 7.2 of \cite{MKS14} showed that the PIN model on the complete graph on $m$ vertices, $K_m$, is Type $\cS$. Using Proposition~\ref{prop:min}, it can be easily verified that the Harary graph PIN model (see \cite{KMS13}), which contains the complete graph PIN model and the PIN model on the $m$-cycle as subclasses, is Type $\cS$. 
\end{remark*}

\section{Are all Type $\cS$ sources $R_{\SK}$-maximal?} \label{sec:counterex}

Section \ref{sec:uh} showed that Type $\cS$ PIN models are $R_{\SK}$-maximal. A natural question that arises is whether all Type $\cS$ sources are $R_{\SK}$-maximal. The answer turns out to be ``No" as seen in the following counterexample.

\begin{example}
Let $W$ be a Ber($p$) rv, for some $p \in [0,1]$: $\Pr[W = 1] = 1 - \Pr[W=0] = p$. Let $X_1,\ldots,X_m$ be rvs that are conditionally independent given $W$, with 
$$\Pr[X_i = 01 |  W = 0] = 1 -  \Pr[X_i = 00 |  W = 0] = 0.5$$ 
and 
$$\Pr[X_i = 11 |  W = 1] = 1 - \Pr[X_i = 10 |  W = 1] = 0.5$$
for $i = 1,2,\ldots,m$. Denote by $h(p)$ the binary entropy of $p$.

It is easy to check that $H(X_A)={|A|}+h(p)$  for all $A\subseteq\cM$, and $H(X_i | X_j) = 1$ for all distinct $i,j \in\cM$. Therefore, all partitions $\cP$ of $\cM$ satisfy $\Delta(\cP)=h(p)$, and hence, $\textbf{I}(X_{\cM})=h(p)$. In particular, $X_{\cM}$ defines a Type $\cS$ source. Furthermore, using \eqref{omni}, we have $R_{\CO}=m$. 

We now show that $R_{\SK}<R_{\CO}$. Consider a Slepian-Wolf code (see \cite[Section 10.3.2]{ElK11}) of rate $H(X_1| X_2)=1$ for terminal 1. All terminals can recover $X_1^n$ since $H(X_1 | X_i)=1$ for all $i\in \{2,3,\cdots,m\}$. Then, using the balanced coloring lemma \cite[Lemma B3]{CN04} on $X_1^n$, an SK of rate $H(X_1)-H(X_1| X_2)=h(p)$ can be obtained. Hence, $R_{\SK}\leq 1<m=R_{\CO}$.
\label{ex:omni}
\end{example}

In fact, there exist non $R_{\SK}$-maximal sources with $\cS$ being a \emph{unique} minimizer for \eqref{eq:I}. To construct such a source we need to define ``clubbing together" of independent multiterminal sources on $\cM$. Formally for independent sources $X_{\cM}^n$ and $Y_{\cM}^n$ define the \emph{clubbed} source $Z_{\cM}^n$ as $Z_i^n=(X_i^n,Y_i^n)$, for all $i\in\cM$. $\Pi_X^*$ and $\Pi_Y^*$ are defined to be the sets of partitions of $\cM$ which are minimizers of \eqref{eq:I} for $X_{\cM}^n$ and $Y_{\cM}^n$ respectively. We will denote the communication complexity (resp. minimum rate of communication for omniscience) for the individual sources $X_{\cM}^n$ and $Y_{\cM}^n$ by $R_{\SK_X}$ and $R_{\SK_Y}$ (resp. $R_{\CO_X}$ and $R_{\CO_Y}$) respectively. The clubbed source satisfies the following result.
\begin{proposition}
Consider two independent multiterminal sources $X_{\cM}^n$ and $Y_{\cM}^n$ and the corresponding clubbed source $Z_{\cM}^n$. Then we have
\begin{equation}
\textbf{I}(Z_{\cM})\geq\textbf{I}(X_{\cM})+\textbf{I}(Y_{\cM}) \label{th:club:1}
\end{equation}
with equality iff $\Pi_X^*\bigcap\Pi_Y^*\neq\emptyset$.
\label{prop:club}
\end{proposition}

\begin{IEEEproof}
Consider any partition $\cP=\{A_1,A_2,\cdots,A_{\ell}\}$ of $\cM$. We have
\begin{align}
\Delta(\cP)&=\frac{1}{\ell-1}\left[\sum_{i=1}^{\ell}H(Z_{A_i})-H(Z_{\cM})\right] \notag\\
                 &=\underbrace{\frac{1}{\ell-1}\left[\sum_{i=1}^{\ell}H(X_{A_i})-H(X_{\cM})\right]}_{\Delta_X(\cP)}\notag\\
                 &\hspace{1.2em}+\underbrace{\frac{1}{\ell-1}\left[\sum_{i=1}^{\ell}H(Y_{A_i})-H(Y_{\cM})\right]}_{\Delta_Y(\cP)} \label{club:1} 
\end{align}
where \eqref{club:1} follows from the independence of $X_{\cM}^n$ and $Y_{\cM}^n$. 

Thus we have from \eqref{club:1} that $\min_{\cP}\Delta(\cP)\geq\min_{\cP}\Delta_X(\cP)+\min_{\cP}\Delta_Y(\cP)$ with equality iff $\cP\in\Pi_X^*\bigcap\Pi_Y^*$. The result follows.
\end{IEEEproof}

We conclude the section by constructing a non $R_{\SK}$-maximal source with $\cS$ being the unique minimizer in \eqref{eq:I}.

\begin{example}
Consider a clubbed source $Z_{\cM}^n=(X_{\cM}^n,Y_{\cM}^n)$, where $X_{\cM}^n$ is the source described in Example \ref{ex:omni} and $Y_{\cM}^n$ corresponds to the PIN model on the complete graph. So, by Lemma \ref{lem:comuh}, we have $\Pi_Y^*=\{\cS\}$. Also, Theorem \ref{th:uh} shows that $Y_{\cM}^n$ is $R_{\SK}$-maximal. 

Since $\Pi_X^*\bigcap\Pi_Y^*=\{\cS\}$, Proposition \ref{prop:club} ensures that independently running protocols achieving $R_{\SK_X}$ and $R_{\SK_Y}$, the SK capacity of $Z_{\cM}^n$ is attained. Also, \eqref{omni} and independence of $X_{\cM}^n$ and $Y_{\cM}^n$ show that $R_{\CO}=R_{\CO_X}+R_{\CO_Y}$. Therefore, $R_{\SK_X}<R_{\CO_X}$ (using Example \ref{ex:omni}) implies that $R_{\SK}<R_{\CO}$.
\label{ex:unique}
\end{example}

\section{Concluding Remarks} \label{sec:conc}

The result of Theorem~\ref{th:uh} is the first exact computation of the communication complexity $R_{\SK}$ in a multiterminal source model with $m > 2$ terminals. In general, however, finding computable expressions or bounds for $R_{SK}$ in a multiterminal setting beyond PIN models appears to be a difficult problem. On the other hand, a more tractable problem may be that of finding a reasonable characterization of the instances of the multiterminal source model which are $R_{\SK}$-maximal. This seems within reach at least for the class of PIN models. For example, one ought to be able to answer the question of whether the Type $\cS$ condition is necessary for (uniform) hypergraph PIN models to be $R_{\SK}$-maximal.

\newpage

\section*{Appendix: Proof of Lemma \ref{lem:mi}}

First we state two lemmas which we will require for the proof. 

\begin{lemma}
\label{lem:rv:2}
For independent random variables $X$,$Y$ and $W$, and any other random variable $Z$, we have
$$
I(X;Z|W)\leq I(X;Z|W,Y).
$$
\end{lemma}
\begin{proof}
This follows by expanding $I(X;Y,Z \mid W)$ in two different ways using the chain rule, and noting that $I(X;Y | W) = 0$.
\end{proof}

\begin{lemma}
\label{lem:rv:1}
For independent random variables $X$ and $Y$, and any other random variable $Z$, we have
$$
I(X;Z)+I(Y;Z)\leq I(X,Y;Z).
$$
\end{lemma}
\begin{proof}
By Lemma~\ref{lem:rv:2}, we have $I(X;Z) \le I(X;Z | Y)$, and hence, $I(X;Z)+I(Y;Z)\leq I(X;Z|Y) + I(Y;Z) = I(X,Y;Z)$.
\end{proof}

\medskip

We begin the proof of Lemma~\ref{lem:mi} by arguing that it is enough to prove the lemma for the PIN model defined by the complete $t$-uniform hypergraph $K_{m,t}$. Consider any hypergraph $H=(\cV,\cE)$ with $|\cV|=m$, and fix a function $\textbf{L}$ of $X_{\mathcal{M}}^n$. Now construct a new source $\tilde{X}_{\cM}^n$ as follows: first consider the set of all $t$-subsets (i.e., subsets of size $t$) of $\cV$ which do not belong in $\cE$, and call it $\cE^c$. Associate with each such $t$-subset $\tilde{e}\in\cE^c$ $n$ i.i.d.\ Ber(1/2) random variables $\tilde{\xi}_{\tilde{e}}^n$. The random variables $\tilde{\xi}_{\tilde{e}}^n$ are assumed to be independent of each other and independent of those associated with the hyperedges in $\cE$. The new source $\tilde{X}_{\cM}^n$ is defined by $\tilde{X}_i^n=(X_i^n, \{\tilde{\xi}_{\tilde{e}}^n: i\in\tilde{e}, \tilde{e}\in\cE^c\})$, for all $i\in\cM$. Observe that the source $\tilde{X}_{\cM}^n$ corresponds to the PIN model on $K_{m,t}$. Moreover, we clearly have
$$
\sum_{i=1}^m I(\tilde{X}_i^n;\textbf{L})\geq \sum_{i=1}^m I(X_i^n;\textbf{L}).
$$ 
Hence it is enough to show that \eqref{eq:mi} holds for the PIN model on $K_{m,t}$.

For the rest of proof we will consider the hypergraph $K_{m,t}$ only. We will also use $X_{\cM}^n$ to denote the source described on $K_{m,t}$. We also have $I(X_{\cM}^n;\textbf{L})=H(\textbf{L})$ from the fact that $\textbf{L}$ is a function of $X_{\cM}^n$. To complete the proof of Lemma \ref{lem:mi}, we will show that the PIN model on $K_{m,t}$ satisfies
\begin{equation}
\sum_{i=1}^m I((\xi_e^n: i\in e, e\in\cE);\textbf{L})\leq t \, I((\xi_e^n: e\in\cE);\textbf{L}). \label{eq:mi1} 
\end{equation}

For any $i\in\cM$, let $\cE_i$ denote the set of hyperedges containing $i$, so that the left-hand side of \eqref{eq:mi1} can be expressed as $\sum_{i=1}^m I\bigl((\xi_e^n:e\in\cE_{i});\textbf{L}\bigr)$. Now, we write $\cE_i$ as a union of two disjoint sets $\cE_{\geq i}$ and $\cE_{\ngtr i}$, i.e., $\cE_i=\cE_{\geq i}\mathop{\dot{\bigcup}}\cE_{\ngtr i}$. The set $\cE_{\geq i}$ is the subset of $\cE_i$ containing no terminals from $\{1,2,\ldots,i-1\}$. The set $\cE_{\ngtr i}$ is thus the subset of  $\cE_i$ containing at least one terminal from $\{1,2,\ldots,i-1\}$. Observe that we have $|\cE_{\geq i}|=\binom{m-i}{t-1}$ for $1\leq i\leq m-t+1$  and $|\cE_{\geq i}|=0$ for $m-t+2\leq i\leq m$. Therefore,
\begin{align}
\sum_{i=1}^m & I\bigl((\xi_e^n:e\in\cE_{i});\textbf{L}\bigr) \nonumber \\
&= I\left(\left(\xi_e^n:e\in\cE_{>1}\right);\textbf{L}\right) \nonumber \\
& \hspace{1.2em} + \sum_{i=2}^{m-t+1}\biggl[I\left(\left(\xi_e^n:e\in\cE_{\ngtr i}\right);\BL\right)\nonumber \\
& \hspace{5em} +I\left(\left(\xi_e^n:e\in\cE_{\geq i}\right);\BL\Big|\left(\xi_e^n:e\in\cE_{\ngtr i}\right)\right)\biggr] \nonumber \\
& \hspace{1.2em} +\sum_{i=m-t+2}^{m} I\left(\left(\xi_e^n:e\in\cE_{i}\right);\textbf{L}\right) \nonumber \\
& \leq  I\left(\left(\xi_{e}^n:e\in\cE_{>1}\right);\textbf{L}\right) \nonumber \\
& \hspace{1.2em}+\sum_{i=2}^{m-t+1} I\left(\left(\xi_e^n:e\in\cE_{\geq i}\right);\BL\Big|\biggl(\xi_e^n:e\in\bigcup_{j\leq i}\cE_{\ngtr j}\biggr)\right)\nonumber \\
& \hspace{1.2em} +\sum_{i=2}^{m-t+1}I\left(\left(\xi_e^n:e\in\cE_{\ngtr i}\right);\BL\right) \nonumber\\
& \hspace{1.2em}+\sum_{i=m-t+2}^{m}I\left(\left(\xi_e^n:e\in\cE_{i}\right);\textbf{L}\right) \label{lem:mi:1} \\
& = \underbrace{I\left(\left(\xi_e^n:e\in\cE\right);\BL\right)}_{P}+\underbrace{\sum_{i=2}^{m-t+1}I\left(\left(\xi_e^n:e\in\cE_{\ngtr i}\right);\BL\right)}_{Q} \nonumber\\
&\hspace{1.2em}+\underbrace{\sum_{i=m-t+2}^m I\left(\left(\xi_e^n:e\in\cE_{i}\right);\textbf{L}\right)}_{R} \label{lem:mi:2} 
\end{align}
where \eqref{lem:mi:1} follows from Lemma \ref{lem:rv:2}. Note that for $t=2$, \eqref{eq:mi1} follows directly from \eqref{lem:mi:2}:  by virtue of Lemma~\ref{lem:rv:1}, we have $Q + R \le P$, so that the right-hand side (RHS) of \eqref{lem:mi:2} is at most $2P$, as desired.  However, the case of $t>2$ is not as simple and needs further work.

To achieve the RHS of \eqref{eq:mi1}, we require $Q+R \le (t-1)P$. We proceed by defining $Q(i)=I\left(\left(\xi_e^n:e\in\cE_{\ngtr i}\right);\BL\right)$ for all $2\leq i\leq m-t+1$, and thus, $Q=\sum_{i=2}^{m-t+1}Q(i)$. Similarly, define $R(i)=I\left(\left(\xi_e^n: e\in\cE_{i}\right);\textbf{L}\right)$ for all $m-t+2\leq i\leq m$, so that $R=\sum_{i=m-t+2}^m R(i)$. The key ideas are the following: 
\begin{enumerate}
\item Expand each $Q(i)$ using the chain rule into conditional mutual information terms of the form $I(\xi_e^n;\BL|\cdots)$, and further condition them on additional $\xi_{\tilde{e}}^n$s appropriately.
\item Allocate these conditional mutual information terms to appropriate $R(i)$s.
\item Use the chain rule to sum each $R(i)$ and the terms allocated to it to obtain $P$. 
\end{enumerate}
Since the conditional mutual information term $I(\xi_e^n;\BL|\cdots)$ can only increase upon further conditioning on additional $\xi_{\tilde{e}}^n$s (by Lemma \ref{lem:rv:2}), we have $Q+R\leq (t-1)P$ as required.

To proceed, we need to define a total ordering on the set $\cE$. We represent a hyperedge $e$ as a $t$-tuple $(i_1i_2\ldots i_t)$, with the $i_j$s, $1\leq j\leq t$, being the terminals which are contained in $e$, ordered according to $i_1<i_2<\ldots<i_t$. Define a total ordering `$<$' on the set $\cE$, `$<$' being the lexicographic ordering of the $t$-tuples. Also based on the ordering `$<$', we index the hyperedges of $\cE$ as $e_j$, $1\leq j\leq\binom{m}{t}$, satisfying $e_i<e_j$ iff $i<j$. As an example, Table \ref{tab:order} illustrates the indexing of the hyperedges in $K_{5,3}$.

\begin{table}[ht!]
\caption{Indexing of the hyperedges in $K_{5,3}$}
\label{tab:order}
\begin{center} 
\small 
\begin{tabular}{|c|c|} \hline
Hyperedge & Index\\\hline
$(123)$ & 1 \\\hline
$(124)$ & 2 \\\hline
$(125)$ & 3 \\\hline
$(134)$ & 4 \\\hline
$(135)$ & 5 \\\hline
$(145)$ & 6 \\\hline
$(234)$ & 7 \\\hline
$(235)$ & 8 \\\hline
$(245)$ & 9 \\\hline
$(345)$ & 10 \\\hline
\end{tabular}
\end{center}
\end{table}

To proceed further, using the chain rule we expand each $Q(i)$ into a sum of conditional mutual information terms of the form $Q_e\triangleq I(\xi_e^n;\BL|(\xi_{\tilde{e}}^n:\tilde{e}<e,\tilde{e}\in\cE))$ as follows:
\begin{align}
Q(i)&=I((\xi_e^n:e\in\cE_{\ngtr i});\BL) \notag \\
      &=\sum_{e\in\cE_{\ngtr i}} I(\xi_e^n;\BL|(\xi_{\tilde{e}}^n:\tilde{e}<e,\tilde{e}\in\cE_{\ngtr i})) \notag \\
      &\leq\sum_{e\in\cE_{\ngtr i}} I(\xi_e^n;\BL|(\xi_{\tilde{e}}^n:\tilde{e}<e,\tilde{e}\in\cE)) \label{lem:mi:3} \\
      &=\sum_{e\in\cE_{\ngtr i}}Q_e \label{lem:mi:4}
\end{align}
where \eqref{lem:mi:3} follows from Lemma \ref{lem:rv:2}. Hence, we have $Q\leq\sum_{i=2}^{m-t+1}\sum_{e\in\cE_{\ngtr i}}Q_e$. A total of $\sum_{i=2}^{m-t+2}\biggl[\binom{m-1}{t-1}-\binom{m-i}{t-1}\biggr]=(t-1)\binom{m-1}{t}$ $Q_e$ terms are generated. Next, each $R(i)$ is allocated $\binom{m-1}{t}$ terms $Q_{e_j}$, $1 \le j \le \binom{m}{t}$, satisfying $i \notin e_j$. This allocation procedure is explained in detail below and is also formalized in Algorithm~\ref{alg:ta}. We add a further conditioning on each $Q_{e_j}$ allocated to $R(i)$ to make it $Q_{e_{j|i}}\triangleq I(\xi_{e_j}^n;\BL|(\xi_{\tilde{e}}^n:\tilde{e}<e_j,\tilde{e}\in\cE),(\xi_{\tilde{e}}^n:\tilde{e}\in\cE_{i}))$. Lemma~\ref{lem:rv:2} and the definition of $Q_{e_{j|i}}$ ensure that $R(i)+\sum_{j:i\notin e_j}Q_{e_j}\leq R(i)+\sum_{j:i\notin e_j}Q_{e_{j|i}}=P$.

We now give a more detailed description of the allocation procedure. Construct a table $T$ with rows indexed by $i = 2,3, \ldots, m-t+1$ and the columns indexed by $j = 1,2,\ldots,\binom{m}{t}$. This table records the availability (for allocation) of a $Q_{e_j}$ from the expansion of $Q(i)$ in \eqref{lem:mi:4}. Initialize the table as follows: $T(i,j)=1$ if a $Q_{e_j}$ came from $Q(i)$ in \eqref{lem:mi:4}; else $T(i,j)=0$. We carry out the allocation procedure on each $R(i)$ in ascending order of $i$. The procedure of allocation is as follows. The idea is to allocate the necessary $Q_{e_j}$s to $R(i)$ in ascending order of $j$. Once an $i$ and $e_j$ are fixed, we test whether $i\notin e_j$ is satisfied. If not, we increment $j$ by 1. If  $i\notin e_j$ is satisfied, then the availability of $Q_{e_j}$ from $Q(k)$, for all $2\leq k\leq m-t+1$, is checked using the table $T$. The smallest $k$ which satisfies $T(k,j)=1$ is chosen, and $R(i)$ is allocated the $Q_{e_j}$ coming from that $Q(k)$. The table is then updated with $T(k,j)=0$ to record that the $Q_{e_j}$ from that $Q(k)$ is no longer available for allocation. We then increment $j$ by 1 and repeat the allocation procedure. Once all $Q_{e_j}$s with $i\notin e_j$ have been allocated to $R(i)$, we begin the allocation procedure for $R(i+1)$. We formally summarize this allocation procedure in Algorithm~\ref{alg:ta}.
\begin{algorithm}
\caption{}
\label{alg:ta}
\begin{algorithmic}
\State $i=m-t+2,j=1$.
\While{$i\leq m, j\leq\binom{m}{t}$}
\If{$i\notin e_j$}
\State $k=2$.
\While{$k\leq m-t+1$}
\If{$T(k,j)=1$}
\State \hspace{-0.5cm}Choose the $Q_{e_j}$ coming from $Q(k)$ in \eqref{lem:mi:4}.
\State \hspace{-0.5cm}Add the additional conditioning to make it $Q_{e_{j|i}}$.
\State \hspace{-0.5cm}Allocate this term to $R(i).$
\State \hspace{-0.5cm}$T(k,j)\gets 0$.
\State \hspace{-0.5cm}Break.
\EndIf
\If{$T(k,j)=0$ \&\& $k=m-t+1$} 
\State Declare ERROR and halt.
\EndIf
\State $k\gets k+1$.
\EndWhile
\EndIf
\State $j\gets j+1$.
\If{$j=\binom{m}{t}+1$}
\State $i\gets i+1$.
\State $j\gets 1$.
\EndIf
\EndWhile
\end{algorithmic}
\end{algorithm}

The flow of Algorithm \ref{alg:ta} for $K_{5,3}$ is illustrated in Example \ref{ex:53} further below. We now make the following claims:
\begin{claim}
\label{cl:1}
Algorithm \ref{alg:ta} never terminates in ERROR.
\end{claim}
\begin{claim}
\label{cl:2}
Algorithm \ref{alg:ta} exhausts all the $Q_e$ terms generated in \eqref{lem:mi:4}.
\end{claim}

Claim~\ref{cl:1} ensures that each $R(i)$, for all $m-t+2\leq i\leq m$, is allocated all the $Q_{e_j}$s satisfying $i\notin e_j$. Therefore, using Claim~\ref{cl:2}, we have 
\begin{align*}
Q+R \ &= \ \sum_{i=m-t+2}^m \left[R(i)+\sum_{j:i\notin e_j}Q_{e_j}\right] \notag\\
& \ \leq\sum_{i=m-t+2}^m \left[R(i)+\sum_{j:i\notin e_j}Q_{e_{j|i}}\right] \ = \ (t-1)P. \notag
\end{align*}
This completes the proof of Lemma~\ref{lem:mi}, modulo the proofs of Claims \ref{cl:1} and \ref{cl:2}, which we give below.

\begin{IEEEproof}[Proof of Claim \ref{cl:1}]
ERROR is possible only if for some $m-t+2\leq i\leq m$ and for some $e$ satisfying $i\notin e$, all the $Q_e$ terms generated in \eqref{lem:mi:4} have already been allocated. This is impossible as there are always enough $Q_e$s. To see this, suppose $e$ contains $t-1-p$ terminals from $\{m-t+2,\ldots,m\}$, i.e., there are $p$ $R(i)$s requiring an allocation of $Q_e$. Since the hypergraph is $t$-uniform, $e$ must contain $p+1$ terminals from $\{1,2,\ldots,m-t+1\}$. This implies that the total number of $Q_e$s generated in \eqref{lem:mi:4} is $p$. Therefore, we clearly have enough $Q_e$s for all $R(i)$s.
\end{IEEEproof}

\begin{IEEEproof}[Proof of Claim \ref{cl:2}]
As discussed earlier, the total number of $Q_e$ terms generated in \eqref{lem:mi:4} is $(t-1)\binom{m-1}{t}$. Also, the total number of $Q_e$ terms required by each $R(i)$ is $\binom{m-1}{t}$. Therefore, using Claim~\ref{cl:1}, the claim follows.
\end{IEEEproof}

\begin{example}
We illustrate how Algorithm~\ref{alg:ta} proceeds for $K_{5,3}$. Denote the hyperedges in $\cE$ using $3$-tuples, i.e., the hyperedge containing terminals $1$, $2$ and $3$ is $(123)$. The indexing of $\cE$ is illustrated in Table~\ref{tab:order}. So for this case we have $Q(2)=I(\xi_{(123)}^n,\xi_{(124)}^n,\xi_{(125)}^n;\BL)$ and $Q(3)=I(\xi_{(123)}^n,\xi_{(134)}^n,\xi_{(135)}^n,\xi_{(234)}^n,\xi_{(235)}^n;\BL)$. Thus, \eqref{lem:mi:4} takes the form
\begin{align}
Q(2) & \leq I(\xi_{(123)}^n;\BL)+I(\xi_{(124)}^n;\BL|(\xi_e^n:e<(124))\notag\\
        &\hspace{1em}+I(\xi_{(125)}^n;\BL|(\xi_e^n:e<(125)) \label{ex:1}\\
Q(3) & \leq I(\xi_{(123)}^n;\BL)+I(\xi_{(134)}^n;\BL|(\xi_e^n:e<(134))\notag\\
        &\hspace{1em}+I(\xi_{(135)}^n;\BL|(\xi_e^n:e<(135))\notag\\
        &\hspace{1em}+I(\xi_{(234)}^n;\BL|(\xi_e^n:e<(234))\notag\\
        &\hspace{1em}+I(\xi_{(235)}^n;\BL|(\xi_e^n:e<(235)) \label{ex:2}
\end{align}
Observe that $R(4)$ and $R(5)$ require four $Q_e$ terms each, and a total of eight $Q_e$ terms are in fact available from \eqref{ex:1} and \eqref{ex:2}. The table $T$ is initialized as follows:
\begin{center} 
\small 
\begin{tabular}{|c||c|c|c|c|c|c|c|c|c|c|} \hline
& 1 & 2 &3 & 4 & 5 & 6 & 7 & 8 & 9 & 10\\\hline\hline
2& 1&1&1&0&0&0&0&0&0&0\\\hline
3& 1&0&0&1&1&0&1&1&0&0\\\hline
\end{tabular}
\end{center}

We will now illustrate a few of the allocations carried out by Algorithm \ref{alg:ta}. The algorithm begins with $i=4$ and $j=1$ and $Q_{(123)}$ needs to be allocated to $R(4)$. With $k=2$ we see that $T(k,1)=1$, and hence we allocate $Q_{(123)}$ coming from $Q(2)$ to $R(4)$. The table $T$ is then updated as below.
\begin{center} 
\small 
\begin{tabular}{|c||c|c|c|c|c|c|c|c|c|c|} \hline
& 1 & 2 &3 & 4 & 5 & 6 & 7 & 8 & 9 & 10\\\hline \hline
2& 0&1&1&0&0&0&0&0&0&0\\\hline
3& 1&0&0&1&1&0&1&1&0&0\\\hline
\end{tabular}
\end{center}

Next we will illustrate the allocation of $Q_{(123)}$ to $R(5)$, i.e., $i=5$ and $j=1$. The state of the table $T$ just before this step is shown below.
\begin{center} 
\small 
\begin{tabular}{|c||c|c|c|c|c|c|c|c|c|c|} \hline
& 1 & 2 &3 & 4 & 5 & 6 & 7 & 8 & 9 & 10\\\hline \hline
2& 0&1&0&0&0&0&0&0&0&0\\\hline
3& 1&0&0&1&0&0&1&0&0&0\\\hline
\end{tabular}
\end{center}

Setting $k=2$, we see that $T(k,1)=0$. So, we move to $k=3$, for which $T(k,1)=1$. Hence the $Q_{(123)}$ term coming from $Q(3)$ is allocated to $R(5)$, and the table $T$ is updated as below.
\begin{center} 
\small 
\begin{tabular}{|c||c|c|c|c|c|c|c|c|c|c|} \hline
& 1 & 2 &3 & 4 & 5 & 6 & 7 & 8 & 9 & 10\\\hline \hline
2& 0&1&0&0&0&0&0&0&0&0\\\hline
3& 0&0&0&1&0&0&1&0&0&0\\\hline
\end{tabular}
\end{center}

We give one last example of an allocation. Observe that $e=(234)$ is the largest (in terms of the ordering on $\cE$) hyperedge such that $Q_e$ needs to be allocated to $R(5)$. We will now illustrate this step. This happens when $i=5$ and $j=7$. The updated table $T$ just before this step is shown below.

\begin{center} 
\small 
\begin{tabular}{|c||c|c|c|c|c|c|c|c|c|c|} \hline
& 1 & 2 &3 & 4 & 5 & 6 & 7 & 8 & 9 & 10\\\hline \hline
2& 0&0&0&0&0&0&0&0&0&0\\\hline
3& 0&0&0&0&0&0&1&0&0&0\\\hline
\end{tabular}
\end{center}

With $k=2$, we see that $T(k,7)=0$. So set $k=3$, and note that $T(k,7)=1$. So, we allocate to $R(5)$ the $Q_{(234)}$ term contributed by $Q(3)$. Upon updating, the table $T$ now has all entries to be $0$. Observe that at this point no other allocation is required, as the $Q_{e_j}$s for $j=8$, $9$ and $10$ are not required by $R(5)$ since terminal $5$ is contained in each of $e_8$, $e_9$ and $e_{10}$. Thus Algorithm \ref{alg:ta} successfully terminates. Finally, we rewrite \eqref{ex:1} and \eqref{ex:2} with underbraces showing the $R(i)$ term to which each $Q_e$ term was allocated by Algorithm \ref{alg:ta}.

\begin{align}
Q(2) & \leq \underbrace{I(\xi_{(123)}^n;\BL)}_{R(4)}+\underbrace{I(\xi_{(124)}^n;\BL|(\xi_e^n:e<(124))}_{R(5)}\notag\\
        &\hspace{1em}+\underbrace{I(\xi_{(125)}^n;\BL|(\xi_e^n:e<(125))}_{R(4)} \label{ex:53:1}\\
Q(3) & \leq \underbrace{I(\xi_{(123)}^n;\BL)}_{R(5)}+\underbrace{I(\xi_{(134)}^n;\BL|(\xi_e^n:e<(134))}_{R(5)}\notag\\
        &\hspace{1em}+\underbrace{I(\xi_{(135)}^n;\BL|(\xi_e^n:e<(135))}_{R(4)}\notag\\
        &\hspace{1em}+\underbrace{I(\xi_{(234)}^n;\BL|(\xi_e^n:e<(234))}_{R(5)}\notag\\
        &\hspace{1em}+\underbrace{I(\xi_{(235)}^n;\BL|(\xi_e^n:e<(235))}_{R(4)} \label{ex:53:2}
\end{align}

It can be clearly seen from \eqref{ex:53:1} and \eqref{ex:53:2} that $R(i), i=4,5,$ have each been allocated with all $Q_e$s with $i\notin e$, and no $Q_e$ is left unallocated.

\label{ex:53}
\end{example}

\end{document}